\documentclass[11pt]{amsart}
\usepackage{mathrsfs}
\usepackage{amssymb} 
\usepackage{url}
\usepackage{hyperref}
\usepackage{graphicx} 

\theoremstyle{definition}
\newtheorem{theorem}{Theorem}[section]

\newtheorem{corollary}[theorem]{Corollary}
\newtheorem{remark}[theorem]{Remark}
\newtheorem{example}[theorem]{Example}

\newtheorem{openproblem}[theorem]{Open Problem}

\def\({\left(}
\def\){\right)}

\newcommand{\de}{\textnormal{d}}

\newcommand{\ds}{\displaystyle}

\newcommand{\eg}{\textit{e.g.} }

\newcommand{\citep}[2]{\cite{#1}, p. #2}

\newcommand{\sref}[1]{\S\ref{#1}}

\newcommand{\image}[3]{\begin{figure*}[ht]
\includegraphics[width=#2\textwidth]{#1}
\caption{\small{\label{#1}#3}}\end{figure*}}

\newcommand{\dsfrac}[2]{\ds{\frac{#1}{#2}}}

\newcommand{\schw}{Schwarzschild}

\def\hyph{-\penalty0\hskip0pt\relax}
\newcommand{\semiriem}{semi{\hyph}Riemannian}

\newcommand{\semireg}{semi{\hyph}regular}
\newcommand{\ssemireg}{Semi{\hyph}regular}

\begin{document} 
 
\title{Schwarzschild's Singularity is Semi-Regularizable}
\author{Cristi \ Stoica}
\date{November 15, 2011}
\thanks{Partially supported by Romanian Government grant PN II Idei 1187.}

\begin{abstract}
It is shown that the Schwarzschild spacetime can be extended so that the metric becomes analytic at the singularity. The singularity continues to exist, but it is made degenerate and smooth, and the infinities are removed by an appropriate choice of coordinates. A family of analytic extensions is found, and one of these extensions is semi-regular. A degenerate singularity doesn't destroy the topology, and when is semi-regular, it allows the field equations to be rewritten in a form which avoids the infinities, as it was shown elsewhere \cite{Sto11a,Sto11b}. In the new coordinates, the Schwarzschild solution extends beyond the singularity. This suggests a possibility that the information is not destroyed in the singularity, and can be restored after the evaporation.
\bigskip
\noindent 
\end{abstract}


\maketitle

\setcounter{tocdepth}{1}
\tableofcontents

\section*{Introduction}


The {\schw} black hole solution, expressed in the {\schw} coordinates, has the following metric tensor:

\begin{equation}
\label{eq_schw_schw}
\de s^2 = -\(1-\dsfrac{2m}{r}\)\de t^2 + \(1-\dsfrac{2m}{r}\)^{-1}\de r^2 + r^2\de\sigma^2,
\end{equation}
where
\begin{equation}
\label{eq_sphere}
\de\sigma^2 = \de\theta^2 + \sin^2\theta \de \phi^2
\end{equation}
is the metric of the unit sphere $S^2$, $m$ the mass of the body, and the units were chosen so that $c=1$ and $G=1$ (see \eg \citep{HE95}{149}).

The first two terms in the right hand side of equation \eqref{eq_schw_schw} don't depend on the coordinates $\theta$ and $\phi$, and $\de\sigma^2$ is independent on the coordinates $r$ and $t$. Therefore we can view this solution as a warped product between a two-dimensional {\semiriem} space and the sphere $S^2$ with the canonical metric \eqref{eq_sphere}. We can use this property to change the coordinates $r$ and $t$ independently, ignoring in calculations the term $r^2\de\sigma^2$, which we reintroduce at the end.

The singularity at $r=2m$, which makes the coefficient $\(1-\dsfrac{2m}{r}\)^{-1}$ become infinite, is only apparent, as shown by the Eddington-Finkelstein coordinates (\citep{HE95}{150}).

As $r\searrow 0$, the coefficient $\(1-\dsfrac{2m}{r}\)^{-1}$ tends to 0, and the coefficient $-\(1-\dsfrac{2m}{r}\)$ tends to $+\infty$. This is a genuine singularity, as we can see from the fact that the scalar $R_{abcd}R^{abcd}$ tends to $\infty$. This seems to suggest that the {\schw} metric cannot be made smooth at $r=0$. In fact, as we will see, we can find coordinate systems in which the components of the metric, although degenerate, are analytic (hence they are finite), even at the genuine singularity given by $r=0$. Moreover, we will see that we can find an analytic extension of the {\schw} spacetime, which is \textit{\semireg}.

In \cite{Sto11a,Sto11b,Sto11d} it was developed the \textit{singular {\semiriem} geometry} for metrics which are allowed to change their signature, in particular to be degenerate. Such metrics $g_{ab}$ are smooth, but $g^{ab}$ tends to $\infty$ when the metric becomes degenerate. The notion of Levi-Civita connection cannot be defined, and the curvature cannot be defined canonically. But in the special case of {\semireg} metrics we can construct a canonical Riemann curvature tensor $R_{abcd}$, which is smooth, although $R^a{}_{bcd}$ is not canonically defined and is singular. It admits canonical Ricci and scalar curvatures, which may be discontinuous or infinite at the points where the metric changes its signature. The usual tensorial and differential operations, normally obstructed by the degeneracy of the metric, can be replaced by equivalent operations which work fine, if the metric is {\semireg}.

In this paper we will show that the {\schw} solution can be extended analytically to such a well behaved {\semireg} solution.

\section{Analytic extension of the {\schw} spacetime}
\label{s_schw_analytic}

\begin{theorem}
\label{thm_schw_analytic}
The {\schw} metric admits an analytic extension at $r=0$.
\end{theorem}
\begin{proof}
It is enough to make the coordinate change in a neighborhood of the singularity -- in the region $r<2m$. On that region, the coordinate $r$ is timelike, and $t$ is spacelike. Let's change the coordinates by
\begin{equation}
\label{eq_coordinate_analytic}
\begin{array}{l}
\bigg\{
\begin{array}{ll}
t &= t(\xi, \tau) \\
r &= r(\xi,\tau) \\
\end{array}
\\
\end{array}
\end{equation}

Recall that the metric coefficients in the {\schw} coordinates are
\begin{equation}
\label{eq_metric_coeff_schw}
g_{tt} = \dsfrac{2m-r}{r},\,
g_{rr} = -\dsfrac{r}{2m-r},\,
g_{tr} = g_{rt} = 0.
\end{equation}
In the new coordinates, the metric coefficients are
\begin{equation}
\label{eq_metric_coeff_analytic_tau_tau}
g_{\tau\tau} = \(\dsfrac{\partial r}{\partial \tau}\)^2g_{rr} + \(\dsfrac{\partial t}{\partial \tau}\)^2 g_{tt}
\end{equation}
\begin{equation}
\label{eq_metric_coeff_analytic_tau_xi}
g_{\tau\xi} = \dsfrac{\partial r}{\partial \tau}\dsfrac{\partial r}{\partial \xi}g_{rr} + \dsfrac{\partial t}{\partial \tau}\dsfrac{\partial t}{\partial \xi}g_{tt}
\end{equation}
\begin{equation}
\label{eq_metric_coeff_analytic_xi_xi}
g_{\tau\tau} = \(\dsfrac{\partial r}{\partial \xi}\)^2g_{rr} + \(\dsfrac{\partial t}{\partial \xi}\)^2 g_{tt}
\end{equation}
From \eqref{eq_metric_coeff_schw} we see that $g_{rr}$ is analytic around $r=0$, hence the only condition for the partial derivatives of $r$ with respect to $\tau$ and $\xi$ is that they are smooth. The expression of $g_{tt}$ on the other hand, has $r$ as denominator, hence we have to cancel it. From equations (\ref{eq_metric_coeff_analytic_tau_tau}--\ref{eq_metric_coeff_analytic_xi_xi}),  we see that $r$ as the denominator in the expression of $g_{tt}$ is canceled if the partial derivatives of $t$ have the form:
\begin{equation}
\label{eq_t_rho}
\begin{array}{l}
\bigg\{
\begin{array}{ll}
	\partial t/\partial \tau &= \rho F_\tau\\
	\partial t/\partial \xi &= \rho F_\xi\\
\end{array}
\\
\end{array}
\end{equation}
where $\rho$, $F_\tau$ and $F_\xi$ are smooth functions in $\tau$ and $\xi$, and
\begin{equation}
\label{eq_r_rho}
r=\rho^2(\tau,\xi).
\end{equation}
The conditions \eqref{eq_t_rho} are satisfied for example if $t$ has the form $t=\xi\rho^2$.

The metric becomes
\begin{equation}
\label{eq_metric_coeff_analytic_tau_tau_rho}
g_{\tau\tau} = -\dsfrac{4\rho^4}{2 m - \rho^2}\(\dsfrac{\partial \rho}{\partial \tau}\)^2 + (2m - \rho^2) F_\tau^2
\end{equation}
\begin{equation}
\label{eq_metric_coeff_analytic_tau_xi_rho}
g_{\tau\xi} = -\dsfrac{4\rho^4}{2 m - \rho^2} \dsfrac{\partial \rho}{\partial \tau}\dsfrac{\partial \rho}{\partial \xi} + (2m - \rho^2) F_\tau F_\xi
\end{equation}
\begin{equation}
\label{eq_metric_coeff_analytic_xi_xi_rho}
g_{\xi\xi} = -\dsfrac{4\rho^4}{2 m - \rho^2}\(\dsfrac{\partial \rho}{\partial \xi}\)^2 + (2m - \rho^2) F_\xi^2
\end{equation}

We can see now that the new metric is smooth around $r=0$ -- none of its components become infinite at $r=0$. It is singular, because it is degenerate, but it is smooth. If the functions $\rho$, $F_\tau$ and $F_\xi$ are analytic, so is the new metric.

To go back to the four-dimensional spacetime, we take the warped product between the above metric and the sphere $S^2$, with the warping function $\rho^2$. The warping function $\rho^2$ is smooth and cancels at the singular points $r=0$. Hence, according to the central theorem of degenerate warped products from \cite{Sto11b}, the warped product between the two-dimensional extension $(\tau,\xi)$ and the sphere $S^2$, with warping function $\rho^2$, is degenerate. This is the needed extension of the {\schw} solution.
\end{proof}

\begin{corollary}
\label{thm_metric_analytic_rho}
The metric in the coordinates $(\tau,\xi,\theta,\phi)$ is
\begin{equation}
\label{eq_schw_analytic_rho}
\de s^2 = -\dsfrac{4\rho^4}{2m-\rho^2}\de \rho^2 + (2m-\rho^2)\(F_\tau\de\tau + F_\xi\de\xi\)^2 + \rho^4\de\sigma^2
\end{equation}
\end{corollary}
\begin{proof}
From \eqref{eq_t_rho} it follows that
\begin{equation}
\de t = \rho(F_\tau\de\tau + F_\xi\de\xi).
\end{equation}
From \eqref{eq_r_rho} it follows that
\begin{equation}
\de r = 2\rho\de\rho.
\end{equation}

By substituting $t$ and $r$ into \eqref{eq_schw_schw} we get the result.
\end{proof}

\begin{corollary}
\label{thm_metric_analytic_det_g_rho}
The determinant of the metric in the new coordinates is
\begin{equation}
\label{eq_metric_analytic_det_g_rho}
\det g = - 4\rho^4 \(F_\tau\dsfrac{\partial \rho}{\partial \xi} - F_\xi\dsfrac{\partial \rho}{\partial \tau}\)^2
\end{equation}
\end{corollary}
\begin{proof}
Direct calculation gives
\begin{equation*}
\det g = g_{tt}g_{rr}
\left|
\begin{array}{ccc}
    \dsfrac{\partial t}{\partial \tau} & \dsfrac{\partial t}{\partial \xi} \\
    \dsfrac{\partial r}{\partial \tau} & \dsfrac{\partial r}{\partial \xi} \\
\end{array}
\right|^2
= -
\left|
\begin{array}{ccc}
    \rho F_\tau & \rho F_\xi \\
    2\rho\dsfrac{\partial \rho}{\partial \tau} & 2\rho\dsfrac{\partial \rho}{\partial \xi} \\
\end{array}
\right|^2
= - 4\rho^4
\left|
\begin{array}{ccc}
    F_\tau & F_\xi \\
    \dsfrac{\partial \rho}{\partial \tau} & \dsfrac{\partial \rho}{\partial \xi} \\
\end{array}
\right|^2
\end{equation*}
\end{proof}

\begin{remark}
The common belief is that it is impossible to extend the {\schw} metric so that it becomes smooth at the singularity, instead of becoming infinite. But this can be done, if we understand that the {\schw} coordinates are singular at $r=0$ (and so are the other known coordinate systems for the {\schw} metric). To pass to a regular coordinate system from a singular one, we need to use a coordinate change which is singular, and the singularity in the coordinate change coincides with the singularity of the metric. This can be viewed as analogous to the Eddington-Finkelstein coordinate change, which removes the apparent singularity on the event horizon. In both cases the metric is made smooth at points where it was thought to be infinite, the only difference is that in our case, at $r=0$, the metric becomes degenerate.
\end{remark}

\begin{example}
\label{thm_metric_analytic_r}
The function $\rho$ has the simplest expression when it depends on only one of the two variables, say $\tau$. To find a coordinate change as we want, let's assume that $\rho$ is the simplest function of $\tau$, $\rho=\tau$.

The metric becomes
\begin{equation}
\label{eq_metric_coeff_analytic_tau_tau_tau}
g_{\tau\tau} = -\dsfrac{4\tau^4}{2 m - \tau^2} + (2m - \tau^2) F_\tau^2
\end{equation}
\begin{equation}
\label{eq_metric_coeff_analytic_tau_xi_tau}
g_{\tau\xi} = (2m - \tau^2) F_\tau F_\xi
\end{equation}
\begin{equation}
\label{eq_metric_coeff_analytic_xi_xi_tau}
g_{\xi\xi} = (2m - \tau^2) F_\xi^2
\end{equation}
The determinant of the metric becomes
\begin{equation}
\label{eq_metric_analytic_det_g_tau}
\det g = - 4\tau^4 F_\xi^2.
\end{equation}
The four-metric takes the form
\begin{equation}
\label{eq_schw_analytic_tau}
\de s^2 = -\dsfrac{4\tau^4}{2m-\tau^2}\de \tau^2 + (2m-\tau^2)\(F_\tau\de\tau + F_\xi\de\xi\)^2 + \tau^4\de\sigma^2
\end{equation}
\end{example}

\begin{example}
\label{thm_metric_analytic_r_t}
The Example \eqref{thm_metric_analytic_r} simplified the form of $r(\tau,\xi)$. We can, in addition, simplify $t(\tau,\xi)$. Equation \eqref{eq_t_rho} suggests that we take $t$ is a product between a power of $\tau$ and a function of $\tau$ and $\xi$. Let's assume that it has the form $\xi\tau^T$, where $T\geq 2$ in order to satisfy \eqref{eq_t_rho}. Hence,
\begin{equation}
\label{eq_coordinate_semireg}
\begin{array}{l}
\bigg\{
\begin{array}{ll}
r &= \tau^2 \\
t &= \xi\tau^T \\
\end{array}
\\
\end{array}
\end{equation}
Then, we have
\begin{equation}
\label{eq_coordinate_jacobian}
\dsfrac{\partial r}{\partial \tau} = 2\tau,\,
\dsfrac{\partial r}{\partial \xi} = 0,\,
\dsfrac{\partial t}{\partial \tau} = T\xi\tau^{T-1},\,
\dsfrac{\partial t}{\partial \xi} = \tau^T
\end{equation}
and
\begin{equation}
\label{eq_F_tau_xi}
\begin{array}{l}
\bigg\{
\begin{array}{ll}
	F_\tau &= T\xi\tau^{T-2} \\
	F_\xi &= \tau^{T-1} \\
\end{array}
\\
\end{array}
\end{equation}
The metric takes the form
\begin{equation}
\label{eq_metric_coeff_analytic_tau_tau_tau_xi}
g_{\tau\tau} = -\dsfrac{4\tau^4}{2 m - \tau^2} + T^2\xi^2(2m - \tau^2) \tau^{2T-4}
\end{equation}
\begin{equation}
\label{eq_metric_coeff_analytic_tau_xi_tau_xi}
g_{\tau\xi} = T\xi(2m - \tau^2) \tau^{2T-3}
\end{equation}
\begin{equation}
\label{eq_metric_coeff_analytic_xi_xi_tau_xi}
g_{\xi\xi} = (2m - \tau^2) \tau^{2T-2}
\end{equation}
and its determinant
\begin{equation}
\label{eq_metric_analytic_det_g_tau_xi}
\det g = - 4 \tau^{2T+2}.
\end{equation}

The four-metric is
\begin{equation}
\label{eq_schw_analytic_tau_xi}
\de s^2 = -\dsfrac{4\tau^4}{2m-\tau^2}\de \tau^2 + (2m-\tau^2)\tau^{2T-4}\(T\xi\de\tau + \tau\de\xi\)^2 + \tau^4\de\sigma^2
\end{equation}
\end{example}

\begin{remark}
\label{rem_jacobian}
When we pass from one coordinate sistem $(\tau,\xi)$ characterized by $T$ to another $(\tau',\xi')$, characterized by $T'\neq T$, as in \eqref{eq_coordinate_semireg}, the transformation has the Jacobian singular at $r=0$. To check this, we use $r=\tau^2=\tau'^2$. The Jacobian is then
\begin{equation}
J=\left(
\begin{array}{ccc}
    \dsfrac{\partial \tau}{\partial \tau'} & \dsfrac{\partial \tau}{\partial \xi'} \\
    \dsfrac{\partial \xi}{\partial \tau'} & \dsfrac{\partial \xi}{\partial \xi'} \\
\end{array}
\right)
= \left(
\begin{array}{ccc}
    \pm 1 & 0 \\
    0 & \tau^{T'-T} \\
\end{array}
\right),
\end{equation}
and it is singular at $\tau=0$.
This seems to suggest that the different coordinate systems we found in the Example \ref{thm_metric_analytic_r_t} represent distinct solutions. This raises the following open question.
\end{remark}

\begin{openproblem}
Can we find natural conditions ensuring the uniqueness of the analytic extensions of the {\schw} solution at the singularity $\tau=0$? Or can we consider all analytic extensions of this type to be equivalent, via coordinate changes which may be singular?
\end{openproblem}

To support the second possibility, we can make the observation that the Jacobian from the Remark \ref{rem_jacobian} is degenerate when we pass from a coordinate system characterized by $T=2$ to another one with another value of $T$, but the converse is not true.

\section{{\ssemireg} extension of the {\schw} spacetime}
\label{s_schw_semireg}

In section \sref{s_schw_analytic} we found an infinite family of coordinate changes which make the metric smooth. As we shall see now, among these solutions there is one which ensures the {\semireg}ity of the metric.

\begin{theorem}
\label{thm_schw_semireg}
The {\schw} metric admits an analytic extension in which the singularity at $r=0$ is {\semireg}.
\end{theorem}
\begin{proof}
To show that the metric is {\semireg}, it is enough to show that there is a coordinate system in which the products of the form
\begin{equation}
\label{eq_semireg_condition_coord}
g^{st}\Gamma_{abs}\Gamma_{cdt}
\end{equation}
are all smooth \cite{Sto11a}, where $\Gamma_{abc}$ are Christoffel's symbols of the first kind. In a coordinate system in which the metric is smooth, as in \sref{s_schw_analytic}
, Christoffel's symbols of the first kind are also smooth. But the inverse metric $g^{st}$ is not smooth for $r=0$. We will show that the products from the expression \eqref{eq_semireg_condition_coord} are smooth.

We use the solution from the Example \ref{thm_metric_analytic_r_t}, and try to find a value for $T$, so that the metric is {\semireg}.

The inverse of the metric has the coefficients given by $g^{\tau\tau} = g_{\xi\xi}/{\det g}$, $g^{\xi\xi} = g_{\tau\tau}/{\det g}$, and $g^{\tau\xi} = g^{\xi\tau} = -g_{\tau\xi}/{\det g}$. It follows from (\ref{eq_metric_coeff_analytic_tau_tau_tau_xi}--\ref{eq_metric_coeff_analytic_xi_xi_tau_xi}) that

\begin{equation}
\label{eq_metric_inv_coeff_semireg_tau_tau}
g^{\tau\tau} = -\dsfrac{1}{4}(2m - \tau^2)\tau^{-4}
\end{equation}
\begin{equation}
\label{eq_metric_inv_coeff_semireg_tau_xi}
g^{\tau\xi} = \dsfrac{1}{4}T\xi(2m - \tau^2)\tau^{-5}
\end{equation}
\begin{equation}
\label{eq_metric_inv_coeff_semireg_xi_xi}
g^{\xi\xi} = \dsfrac{\tau^{-2T+2}}{2 m - \tau^2} - \dsfrac{1}{4}T^2 \xi^2 (2m - \tau^2)\tau^{-6}
\end{equation}
Christoffel's symbols of the first kind are given by
\begin{equation}
\label{eq_christoffel}
\Gamma_{abc} = \dsfrac 1 2 \(\partial_a g_{bc} + \partial_b g_{ca} - \partial_c g_{ab}\),
\end{equation}
so we have to calculate the partial derivatives of the coefficients of the metric.

From \eqref{eq_coordinate_jacobian} and (\ref{eq_metric_coeff_analytic_tau_tau_tau_xi}--\ref{eq_metric_coeff_analytic_xi_xi_tau_xi}) we have:
\begin{equation*}
\begin{array}{lll}
\partial_{\tau}g_{\tau\tau} &=& \partial_{\tau}\(-\dsfrac{4\tau^4}{2 m - \tau^2} + \xi^2T^2(2m - \tau^2) \tau^{2T-4}\) \\
&=& -4\dsfrac{4\tau^3(2 m - \tau^2) + 2\tau^5}{(2 m - \tau^2)^2} + 2T^2(2T-4) m \xi^2\tau^{2T-5} \\
&&- T^2 (2T-2)\xi^2\tau^{2T-3}, \\
\end{array}
\end{equation*}
hence
\begin{equation}
\label{eq_pd_tau_tau_tau}
\partial_{\tau}g_{\tau\tau} = 8\dsfrac{\tau^5 - 4m\tau^3}{(2 m - \tau^2)^2} + 2T^2(2T-4) m \xi^2\tau^{2T-5}- T^2 (2T-2)\xi^2\tau^{2T-3}.
\end{equation}
Similarly,
\begin{equation}
\label{eq_pd_tau_tau_xi}
\partial_{\tau}g_{\tau\xi} = 2 T (2T-3) m\xi\tau^{2T-4} - T (2T-1) \xi\tau^{2T-2},
\end{equation}
\begin{equation}
\label{eq_pd_tau_xi_xi}
\partial_{\tau}g_{\xi\xi} = 2 m (2T-2)\tau^{2T-3} - 2T\tau^{2T-1},
\end{equation}
\begin{equation}
\label{eq_pd_xi_tau_tau}
\partial_{\xi}g_{\tau\tau} = 2 T^2 \xi (2m - \tau^2)\tau^{2T-4},
\end{equation}
\begin{equation}
\label{eq_pd_xi_tau_xi}
\partial_{\xi}g_{\tau\xi} = T(2m - \tau^2)\tau^{2T-3},
\end{equation}
and
\begin{equation}
\label{eq_pd_xi_xi_xi}
\partial_{\xi}g_{\xi\xi} = 0.
\end{equation}

To ensure that the expression \eqref{eq_semireg_condition_coord} is smooth, we use power counting and try to find a value of $T$ for which it doesn't contain negative powers of $\tau$. The least power of $\tau$ in the partial derivatives of the metric is $\min(3,2T-5)$, as we can see by inspecting equations (\ref{eq_pd_tau_tau_tau}--\ref{eq_pd_xi_xi_xi}). The least power of $\tau$ in the inverse metric is $\min(-6,-2T+2)$, as it follows from the equations (\ref{eq_metric_inv_coeff_semireg_tau_tau}--\ref{eq_metric_inv_coeff_semireg_xi_xi}). Since $\min(-6,-2T+2) = -3 - \max(3,2T-5)$, the condition that the least power of $\tau$ in \eqref{eq_semireg_condition_coord} is non-negative is
\begin{equation}
	-1 - 2T + 3\min(3,2T-5) \geq 0
\end{equation}
with the unique solution
\begin{equation}
	T = 4.
\end{equation}
Hence, taking $T=4$ ensures the smoothness of \eqref{eq_semireg_condition_coord}, and by this, the {\semireg}ity of the metric in two dimensions $(\tau,\xi)$.

When going back to four dimensions, we remember the central theorem of {\semireg} warped products from \cite{Sto11b}, stating that the warped product between the two-dimensional extension $(\tau,\xi)$ and the sphere $S^2$, with warping function $\tau^2$, is {\semireg}.
\end{proof}

It is useful to extract from the proof the expression of the metric:

\begin{corollary}
The metric
\begin{equation}
\label{eq_schw_semireg}
\de s^2 = -\dsfrac{4\tau^4}{2m-\tau^2}\de \tau^2 + (2m-\tau^2)\tau^4\(4\xi\de\tau + \tau\de\xi\)^2 + \tau^4\de\sigma^2
\end{equation}
is an analytic extension of the {\schw} metric, which is {\semireg}, including at the singularity $r=0$.
\end{corollary}
\qed

\begin{remark}
The Riemann curvature tensor $R_{abcd}$ is smooth, because the metric is {\semireg}. How can it be smooth, when we know that the Riemann curvature of the {\schw} metric tends to infinity when it approaches the singularity $r=0$? The answer is that the coefficients of $R_{abcd}$ depend on the coordinate system. Since the usual coordinates used with the {\schw} black hole solution are singular with respect to ours, a tensor which is smooth  in our coordinates may appear singular in {\schw} coordinates.
But this should not be a big surprise, because for the {\schw} solution the Ricci tensor is $0$, hence the scalar curvature is $0$ too, and the Einstein's equation is simply $T_{ab}=0$.
On the other hand, the Kretschmann invariant $R_{abcd}R^{abcd}$ still becomes infinite at $r=0$, of course, being a scalar, and therefore remaining unchanged at the coordinate transforms. But it does become infinite only because $R^{abcd}$ becomes infinite.
\end{remark}

\section{Penrose-Carter coordinates for the {\semireg} solution}
\label{s_schw_semireg_penrose_carter}

To move to Penrose-Carter coordinates, we apply the same steps as one usually applies for the {\schw} black hole (\citep{HE95}{150-156}). More precisely, the lightlike coordinates for the Penrose-Carter diagram are
\begin{equation}
\label{eq_schw_penrose_coord}
\begin{array}{l}
\Bigg\{
\begin{array}{lll}
v''&=&\arctan\(\phantom{-}(2m)^{-1/2}\exp\(\phantom{-}\dsfrac{v}{4m}\)\) \\
w''&=&\arctan\(-(2m)^{-1/2}\exp\(-\dsfrac{w}{4m}\)\) \\
\end{array}
\\
\end{array}
\end{equation}
where $v,w$ are the Eddington-Finkelstein lightlike coordinates
\begin{equation}
\label{eq_schw_lightlike_coord}
\begin{array}{l}
\bigg\{
\begin{array}{lll}
v &=& t+r+2m\ln(r-2m) \\
w &=& t-r-2m\ln(r-2m). \\
\end{array}
\\
\end{array}
\end{equation}

\image{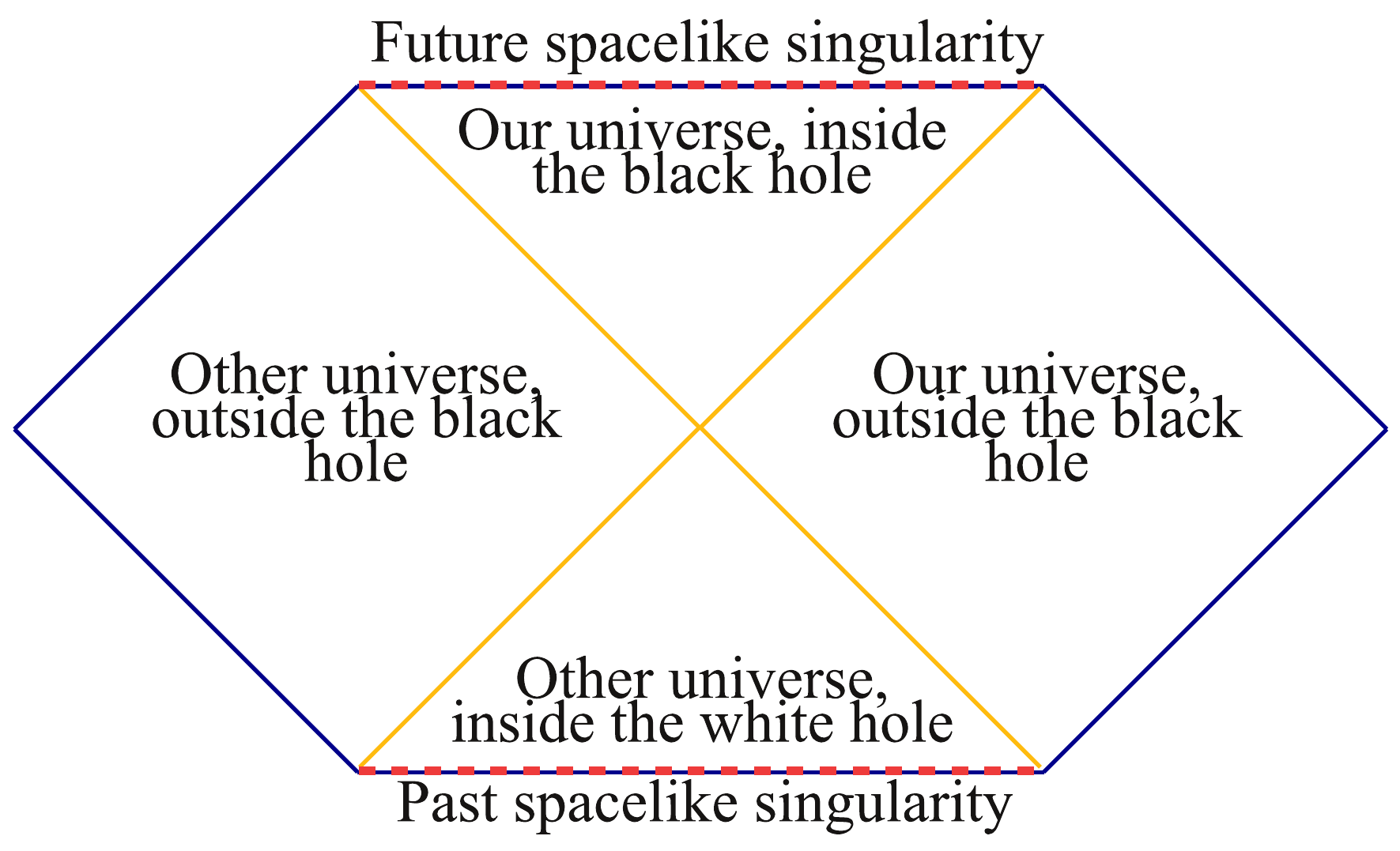}{0.5}{The maximally extended {\schw} solution, in Penrose-Carter coordiantes.}

Usually, in the Penrose-Carter diagram of the {\schw} spacetime is considered that the maximal analytic extension is given by the conditions $v''+w''\in(-\pi,\pi)$ and $v'',w''\in\(-\dsfrac{\pi}{2},\dsfrac{\pi}{2}\)$ (see Fig. \ref{std-schwarzschild}). This is because we have to stop at the singularity $r=0$, because the infinite values we get prevent the analytic continuation.

To move from the coordinates $(\tau,\xi)$ to the Penrose-Carter coordinates, we use the substitution \eqref{eq_coordinate_semireg}:
\begin{equation}
\label{eq_schw_semireg_lightlike_coord}
\begin{array}{l}
\bigg\{
\begin{array}{lll}
v &=& \xi\tau^4 + \tau^2 + 2m\ln(2m - \tau^2) \\
w &=& \xi\tau^4 - \tau^2 - 2m\ln(2m - \tau^2). \\
\end{array}
\\
\end{array}
\end{equation}

Our coordinates allow us to go beyond the singularity. As we can see from equation \eqref{eq_schw_semireg}, our solution extends to negative $\tau$ as well. From \eqref{eq_schw_semireg_lightlike_coord} we see that it is symmetric with respect to the hypersurface $\tau=0$. This leads to the Penrose-Carter diagram from Fig. \ref{semireg-schwarzschild}.

\image{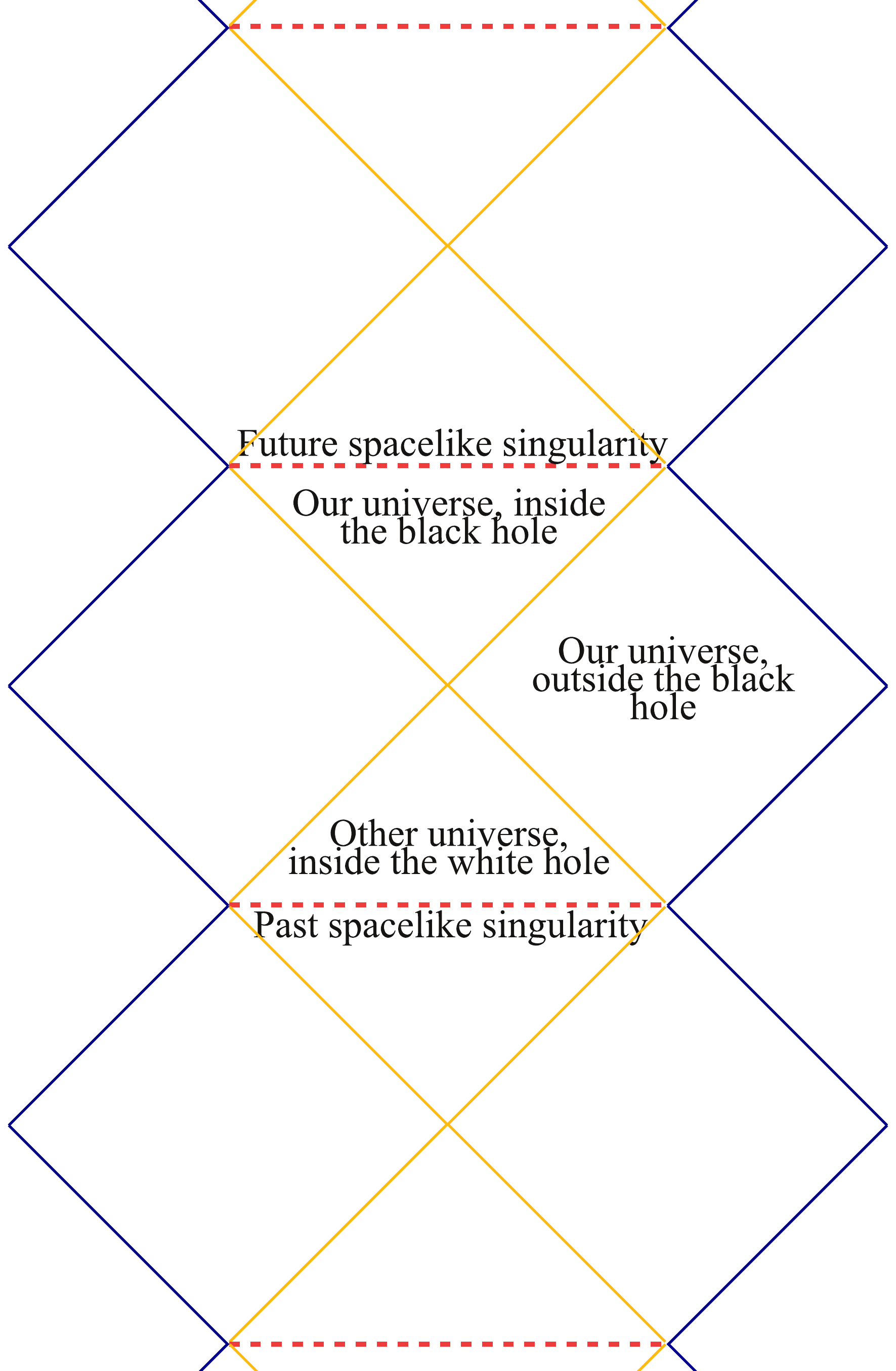}{0.5}{Our maximally extended {\schw} solution, in Penrose-Carter coordiantes.}

\section{The significance of the {\semireg} solution}
\label{s_schw_semireg_significance}

The main consequence of the extensibility of the {\schw} solution to a {\semireg} solution beyond the singularity is that the information is not lost there. This can apply as well to the case of an evaporating black hole (see Fig. \ref{evaporating-bh-s}).

Because of the no-hair theorem, the {\schw} solution is representative for non-rotating and electrically neutral black holes. If the black hole evaporates, the information reaching the singularity is lost (Fig. \ref{evaporating-bh-s} A). If the singularity is {\semireg}, it doesn't destroy the topology of spacetime (Fig. \ref{evaporating-bh-s} B). Moreover, although normally in this case the covariant derivative and other differential operators can't be defined, there is a way to construct them naturally (as shown in \cite{Sto11a}), allowing the rewriting of the field equations for the {\semireg} case, without running into infinities. This ensures that the field equations can go beyond the singularity.

\image{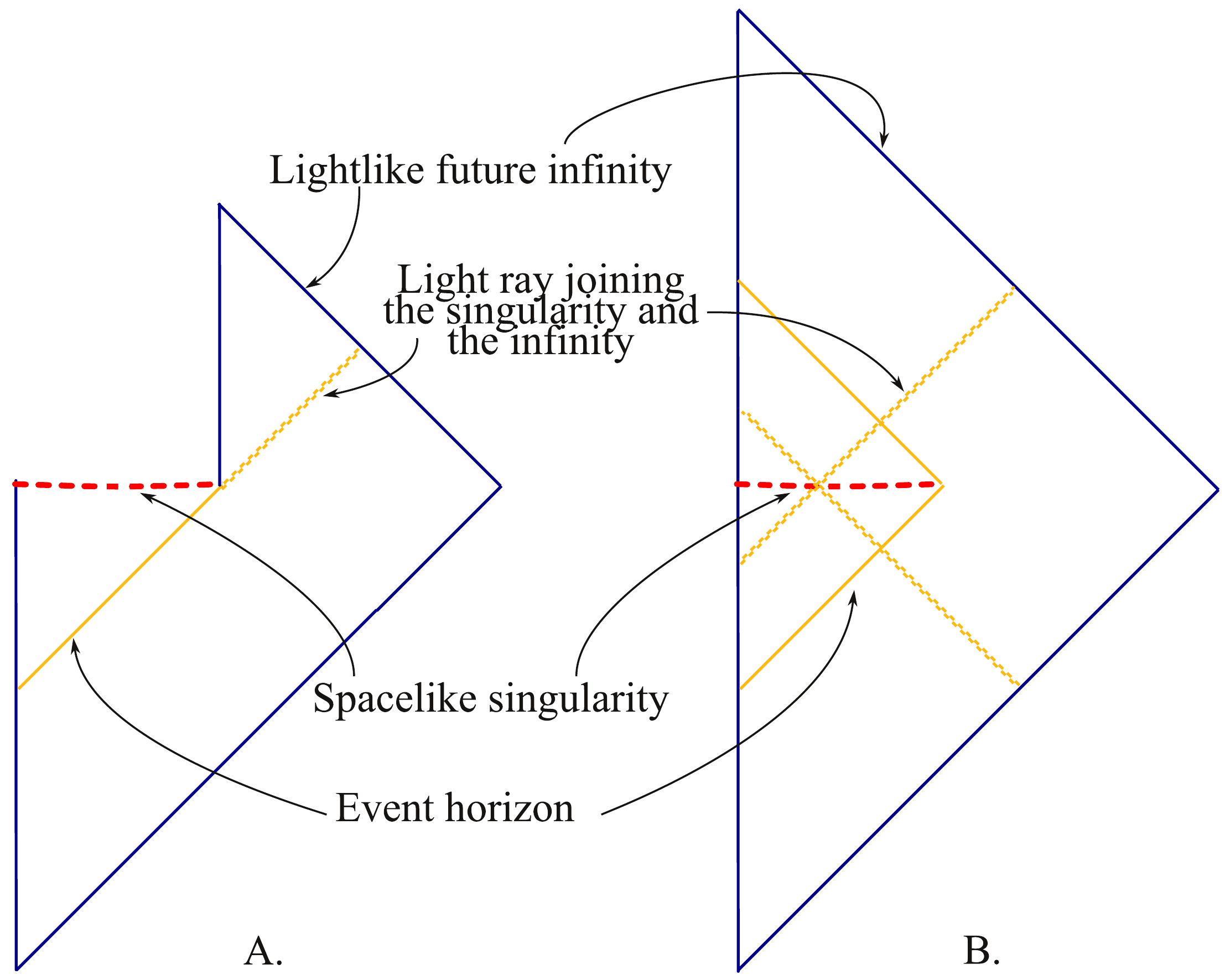}{0.7}{The Penrose-Carter diagram for a non-rotating and electrically neutral evaporating black hole whose singularity destroys the information.
\textbf{A.} Standard evaporating black hole.
\textbf{B.} Evaporating black hole extended beyond the singularity, whose singularity doesn't destroy the information.}

In the case of a black hole which is not primordial and evaporates completely in a finite time, all of the light rays traversing the singularity reach the past and future infinities. This means that the presence of a spacelike evaporating black hole is compatible with the global hyperbolicity, as in the diagram \ref{evaporating-bh-s} B.

The singularity is accompanied by violent and very destructive forces. But, as we can see from the {\semireg} formulation, there is no reason to consider that it destroys the information or the structure of spacetime.

\bibliographystyle{amsplain}
\providecommand{\bysame}{\leavevmode\hbox to3em{\hrulefill}\thinspace}
\providecommand{\MR}{\relax\ifhmode\unskip\space\fi MR }
\providecommand{\MRhref}[2]{%
  \href{http://www.ams.org/mathscinet-getitem?mr=#1}{#2}
}
\providecommand{\href}[2]{#2}

\end{document}